\documentclass[9pt,twocolumn,twoside]{osajnl}

\journal{pr} 

\setboolean{shortarticle}{false}

\usepackage{amsmath}
\usepackage{amsthm}
\usepackage{graphicx}

\newcommand{\cj}{\mathrm{j}}
\newcommand{\e}{\mathrm{e}}

\newtheorem{theorem}{Theorem}

\newtheorem{corollary}{Corollary}
\newtheorem{definition}{Definition}

\title{Provable Routing Analysis of Programmable Photonics}

\author[1,*]{Zhengqi Gao}
\author[2,3]{Xiangfeng Chen}
\author[1]{Zhengxing Zhang}
\author[1]{Chih-Yu Lai}
\author[1]{Uttara Chakraborty}
\author[2,3]{Wim Bogaerts}
\author[1]{Duane S. Boning}

\affil[1]{Microsystems Technology Laboratories, Electrical Engineering and Computer Science, Massachusetts Institute of Technology, Cambridge MA 02139, USA}
\affil[2]{Ghent University - IMEC, Department of Information Technology, Gent, Belgium}
\affil[3]{Center of Nano- and Biophotonics, Ghent University, Gent, Belgium}

\affil[*]{Corresponding author: Zhengqi Gao}

\begin{abstract}
Programmable photonic integrated circuits (PPICs) are an emerging technology recently proposed as an alternative to custom-designed application-specific integrated photonics. Light routing is one of the most important functions that need to be realized on a PPIC. Previous literature has investigated the light routing problem from an algorithmic or experimental perspective, e.g., adopting graph theory to route an optical signal. In this paper, we also focus on the light routing problem, but from a complementary and theoretical perspective, to answer questions about what is possible to be routed. Specifically, we demonstrate that not all path lengths (defined as the number of tunable basic units that an optical signal traverses) can be realized on a square-mesh PPIC, and a rigorous realizability condition is proposed and proved. We further consider multi-path routing, where we provide an analytical expression on path length sum, upper bounds on path length mean/variance, and the maximum number of realizable paths. All of our conclusions are proven mathematically. Illustrative potential optical applications using our observations are also presented.
\end{abstract}

\setboolean{displaycopyright}{true}

\begin{document}

\maketitle

\section{Introduction}\label{sec:intro}
Over the past two decades, photonic integrated circuits (PICs) have been demonstrated in a growing number of applications and fields, such as data communications, quantum computing, and optical beam-steering ~\cite{soref2006silicon,wim_silicon_challenge,shen2017deep,madsen2022quantum,arrazola2021quantum,kim2019opa}. Usually, a PIC is designed for one particular application, which is commonly referred to as application-specific. For every new application, engineers must design and fabricate a new photonic circuit from scratch. Depending on the technology, this cycle can take one year or more. Recently, programmable photonic integrated circuits (PPICs)~\cite{Perez16Reconfigurable,wang2022programmable,Zhuang15programmable,capmany2016programmable,bogaerts2020programmable,perez2017multipurpose,perez2020multipurpose,gao2022automatic1,gao2022automatic2,chen2020graph} have emerged as an alternative paradigm, exploiting reconfigurability to avoid the redesign workload. Specifically, a PPIC is made up of a mesh of so-called \textit{tunable basic units} (TBUs), and each TBU has two actively controlled optical tuners (e.g., electro-optic phase shifters). By tuning these individual actuators, the flow of light in a PPIC can be reconfigured to realize various linear light processing functions, such as splitting, interfering, routing, and filtering~\cite{Perez16Reconfigurable,perez2017multipurpose,perez2020multipurpose}. 

Depending on the connectivity of these TBUs, a PPIC can be categorized into two types: feed-forward and loop-back (recirculating) topologies~\cite{bogaerts2020programmable,gao2022automatic2}. Some special feed-forward topologies~\cite{clements2016optimalforward,reck1994experimental} have been proven capable of realizing any unitary transformations, and thus have become popular as an acceleration engine to implement matrix-vector multiplications for optical neural networks~\cite{shen2017deep}. On the other hand, loop-back topologies~\cite{perez2017multipurpose} have the ability to redirect light in any direction in the circuit, and implement tunable delay lines, interferometric filters, and ring resonators. They are more versatile and can be useful in more optical applications compared to feed-forward topologies. The most common configurations of a loop-back PPIC are a triangular, square, or hexagonal mesh~\cite{Perez16Reconfigurable,bogaerts2020programmable}, and these are also the main focus in our paper.

One of the most important functions for a PPIC to realize is light routing. Some earlier published papers~\cite{chen2020graph,lopez2020auto} model a PPIC using a directed/undirected graph, and solve the routing problem using existing graph algorithms. These efforts are oriented towards answering the question of how to route, while in our paper, we study (and answer) a series of complementary questions related to what can be routed or not. Our analysis is performed based on the metric of path length, which is defined as the number of TBUs in the optical path~\cite{lopez2020auto}. When an optical signal traverses through a PPIC, its phase change and time delay is closely related to its path length. Thus, investigating path length will provide us a good understanding of the routing ability of PPICs and will be instructive in many optical applications. As an example, if a PPIC can route at most $y$ (e.g., $y=5$) optical signals such that their path lengths are all equal, then this indicates that at most $y$ optical signals can go through this PPIC while maintaining their relative phase differences. Such a conclusion is of crucial interest in a phase sensitive optical application. An immediate question is: What is the maximum value of $y$ for a given PPIC? Our paper focuses on answering such questions. All of our findings will be supported by mathematical proofs. Finally, a number of illustrative potential optical applications using our observations are explored.

The paper is organized as follows. In Section~\ref{sec:preliminary}, we briefly review the compact model of the TBU, formally make several definitions (such as floating node and path length), and describe several axiomatic conclusions. Next, in Section~\ref{sec:routing}, we present our major results about the routing ability of a general square mesh, covering both the cases of single-path and multi-path routing. In Section~\ref{sec:applications}, we present a few applications inspired by our findings. Finally, we conclude our paper with Section~\ref{sec:conclusion}. Our main text only focuses on square meshes; analyses of triangular and hexagonal meshes are considered in the Appendix.


\section{Preliminary}\label{sec:preliminary}

For this study, we focus on an $N\times M$ square mesh (see the Appendix for other mesh topologies), made up of TBUs arranged along the horizontal and vertical edges, as shown in Fig.~\ref{fig:undirected_simple_graph_representation}(a-b). The TBU can be implemented in different ways~\cite{bogaerts2020programmable}, but the most common implementation is based on variants of a $2 \times 2$ Mach–Zehnder interferometer (MZI)~\cite{bogaerts2020programmable,perez2017multipurpose}. In this paper, we consider the same TBU implementation as in~\cite{gao2022automatic2}. An MZI has two inputs and two outputs, and their values are related by a transfer matrix. Specifically, the two output optical signals equal the product of the following $2\times2$ transfer matrix with the input optical signals:

\begin{equation}\label{eq:our_compact_S}
        \mathbf{F} =     {0.5\left[
    \begin{array}{cc}
        \e^{-\cj\theta}-\e^{-\cj\phi} &  -\cj\e^{-\cj\theta}-\cj\e^{-\cj\phi} \\
        -\cj\e^{-\cj\theta}-\cj\e^{-\cj\phi} &   -\e^{-\cj\theta}+\e^{-\cj\phi}
    \end{array}\right]} \alpha\e^{-\cj\omega\frac{n_{\text{eff}}L}{c}}
\end{equation}
where $\{\theta,\phi\}$ are the active phase shifts of the TBU which are tuned by electric signals, $\omega$ is the angular frequency of the optical carrier wave, $\alpha$ represents the TBU loss, $n_{\text{eff}}=n_{\text{eff}}(\omega)$ represents the effective index of the propagating mode in the waveguides, $c$ is the free space light speed, and $L$ is the length of the waveguide in the TBU.

There are two special cases of primary interest: (i)~bar state: $\theta=0$ and $\phi=\pi$, and (ii)~cross state: $\theta=\phi=-0.5\pi$. The resulting $\mathbf{F}$ for these two cases are respectively shown below:
\begin{equation}\label{eq:cross_bar}
\begin{aligned}
    &\text{bar state: }\; &\mathbf{F} = \left[
    \begin{array}{cc}
        1 & 0 \\
        0 &  -1
    \end{array}
    \right]\alpha\e^{-\cj\omega\frac{n_{\text{eff}}L}{c}}
    \\
    &\text{cross state: }\; &\mathbf{F} = \left[
    \begin{array}{cc}
        0 & 1 \\
        1 & 0
    \end{array}
    \right]\alpha\e^{-\cj\omega\frac{n_{\text{eff}}L}{c}},\quad 
\end{aligned}
\end{equation}
For a horizontal TBU in the bar state, an optical signal going in from the top left port will be guided to its top right port (i.e., confined in the same arm). Alternatively, for that TBU in the cross state, an optical signal going in from its top left port will be guided to the bottom right port. When a square-mesh PPIC is used merely to route light, all TBUs are either set to cross state or bar state. Thus, following this convention, we assume all TBUs are either in cross state or bar state in our paper, and no partial coupling is allowed. Now, we formally define the graph representation of a square-mesh PPIC, where the ports of the TBUs are represented by nodes.

\begin{figure}[!htb]
    \centering
    \includegraphics[width=1.0\linewidth]{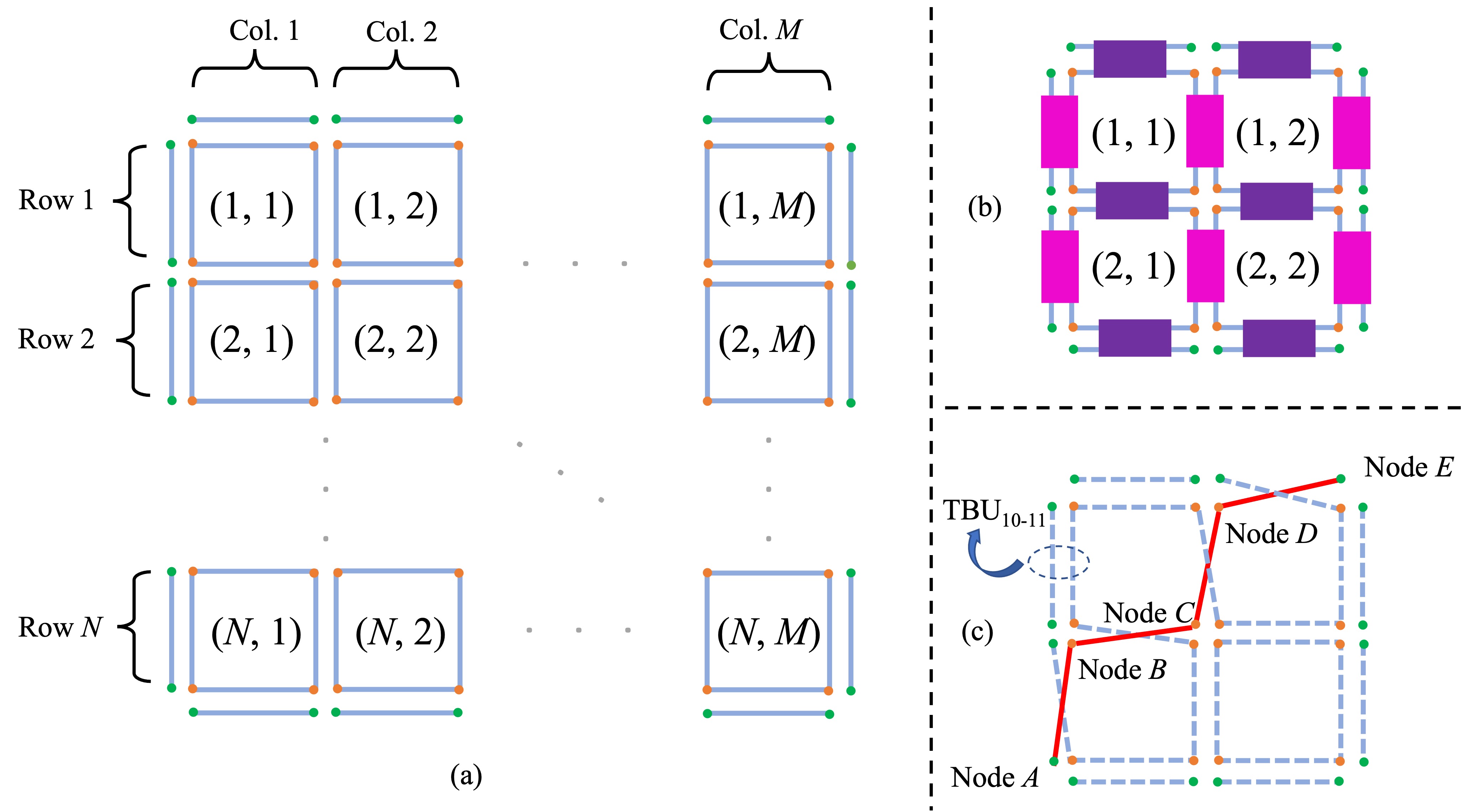}
    \caption{(a) An undirected simple graph representation of an $N\times M$ square-mesh PPIC when all TBUs are in bar state. The green and orange dots represent floating and non-floating nodes, respectively. (b)~An example of a $2\times 2$ square mesh when all TBUs are in bar state. Purple and pink rectangles represent horizontal and vertical TBUs, respectively. (c)~The optical path marked by red has length $4$. In this configuration, four TBUs are in cross state.}
    \label{fig:undirected_simple_graph_representation}
\end{figure}

\begin{definition}
an $N\times M$ square-mesh PPIC can be represented by an undirected simple graph, with parallel/cross line segments for bar/cross state, respectively, as demonstrated in Fig.~\ref{fig:undirected_simple_graph_representation} (a) and (c).
\end{definition}


We will respectively use $N$ and $M$ to represent the number of rows and columns throughout our paper. We will use the format $\text{TBU}_{ij-mn}$ to refer to the TBU at the intersection of square cell $(i,\,j)$ and $(m,\,n)$. To ease the mathematical description later, we introduce two definitions, based on the concepts of `node degree' and `path' in graph theory. 

\begin{definition}
We define a floating node as a node with only one edge connected to it (i.e., node degree equal to $1$).  Similarly, a non-floating node refers to a node with at least two edges connected to it (i.e., node degree no less than $2$).
\end{definition}

\begin{definition}
We define an undirected optical path as a path both starting from and ending at a floating node. Path length is defined as the number of edges (or equivalently, the number of TBUs) that the path passes through.
\end{definition}

Several things need to be clarified. First, in a square mesh, a non-floating node can only have node degree equal to or less than $2$, but not larger than $2$, because at most two edges are connected to a node. Second, due to reciprocity of light propagation in passive circuits, one undirected optical path actually corresponds to two directed light paths. For instance, in Fig.~\ref{fig:undirected_simple_graph_representation}~(c), we will use $(A,B,C,D,E)$ to denote the undirected optical path marked by red of length $4$, while that path actually corresponds to two directed optical paths $A\to B \to C \to D \to E$ and $E \to D \to C \to B \to A$ in potential applications. Equivalently, $(E,D,C,B,A)$ is another valid notation for this undirected light path. However, $(E,D,C,B,A)$ with start at $E$ and end at $A$, and $(A,B,C,D,E)$ with start at $A$ and end at $E$, will be counted as one undirected light path instead of two in our paper. Last, but not least, the metric path length is closely related to the phase change of an optical signal when it goes through a square mesh. Given an optical path of length $l$, assume that a time-harmonic optical signal with input complex magnitude $b$ following this path goes through the square mesh. Then, according to Eq.~(\ref{eq:cross_bar}), the output response is:
\begin{equation}\label{eq:path_length_response}
    b \cdot (-1)^q \cdot (\alpha\e^{-\cj \omega\frac{n_{\text{eff}}L}{c}})^l=b \cdot \alpha^{l}\e^{-\cj (\omega\frac{n_{\text{eff}}lL}{c}+q\pi)}
\end{equation}
Note that the $(-1)^q$ factor corresponds to the fact that a TBU in bar state might introduce an additional $180^{\circ}$ phase shift, as shown in Eq.~(\ref{eq:cross_bar}). Here $q$ can be regarded as either $0$ or $1$, depending on the number of TBUs in bar state in the path. Nevertheless, the $q\pi$ phase shift is trivial, as it can be compensated if we append additional phase shifts at the output of the square mesh. More importantly, what we focus on are the remaining term $l(\omega{n_{\text{eff}}L}/{c})$ which characterizes the phase response, and the term $\alpha^l$ which impacts the magnitude response. These expressions manifest the rationale for basing our analysis on the path length $l$. To facilitate discussion later in the paper, we introduce the concept of peripheral TBU:

\begin{definition}
A peripheral TBU refers to a TBU possessing floating nodes. A non-peripheral TBU is a TBU with all four of its nodes non-floating.
\end{definition}

As shown in Fig.~\ref{fig:undirected_simple_graph_representation}~(c), $\text{TBU}_{10-11}$ is a peripheral TBU placed vertically at the top left position. This example also indicates that to apply our TBU naming format to a peripheral TBU, we have to imagine a notional additional square cell (e.g., cell $(1,\,0)$ in this case) at the outer left side. In the following, we present several axioms in Theorem~\ref{thm:oobvious} to lay the groundwork for later analysis.

\begin{theorem}\label{thm:oobvious}
For an $N\times M$ square mesh, we have the following conclusions:
\begin{enumerate}[label = (\arabic*)]
    \item There are $N(M+1)+M(N+1)$ TBUs in the circuit.
    \item The total number of circuit configurations is $2^{N(M+1)+M(N+1)}$.
    \item The total numbers of floating and non-floating nodes are $(4N+4M)$ and $4NM$, respectively.
    \item For a fixed circuit configuration, the total  number of undirected optical paths is $(2N+2M)$. 
    \item An undirected optical path has two different floating nodes, at the path start and end, respectively.
    \item An undirected optical path passes through a vertical and a horizontal TBU in turns, abbreviated to '$\cdots$VHVH$\cdots$'. 
\end{enumerate}
\end{theorem}


The above conclusions can be proved by straightforward calculation. For~(2), we note that each TBU has two state choices (i.e., either cross or bar), and thus, the total number of circuit configurations equals $2^{\#\text{TBU}}$, where $\#\text{TBU}$ is provided in~(1). For a curious reader, we pose what may appear to be a simple question: can a $2 \times 3$ square mesh implement a path of length $3$ or $7$? After drawing and trying many possibilities by hand, we find that a $2\times 3$ square mesh can implement a path of length $7$, but not $3$. Even more surprisingly, neither is realizable in a $2\times 2$ square mesh. To this end, our first question, which will be answered in the next section, is: \textbf{In an $N\times M$ square mesh, is path length $x$ ($x \in \mathbb{Z}^{+}$) realizable by some circuit configuration?}

\section{Routing Ability of Square Mesh}\label{sec:routing}

\subsection{Realizability of a Single Path}\label{sec:single_path}

First, we attempt to bound the value of $x$ by investigating the minimum and maximum realizable path length in an $N\times M$ square. Obviously, the minimum path length equals $1$. The following Theorem~\ref{thm:maximum_length} answers the maximum length question.

\begin{theorem}\label{thm:maximum_length}
Enumerating all possible circuit configurations, the maximum path length that an $N\times M$ square mesh can achieve equals $(4NM+1)$.
\end{theorem}

\begin{proof}
The proof is made up of two parts: (i)~the maximum path length cannot be larger than $(4NM+1)$, and (ii)~a path of length $(4NM+1)$ is indeed realizable in an $N\times M$ square mesh. For (i), we prove by contradiction: Assume a path of length $x\geq 4MN+2$ is realizable. Then, the number of nodes on this path is $(x+1)$, which is no less than $(4MN+3)$. Since we only have $4NM$ non-floating nodes as depicted in Theorem~\ref{thm:oobvious}~(3), then this path must experience at least $3$ floating nodes, which contradicts Theorem~\ref{thm:oobvious}~(5). 

Now, we prove (ii)~by constructing a path of length $(4NM+1)$. In outline, we first show the construction in the case of a $2\times 2$ square mesh in Fig.~\ref{fig:proof}~(a). The key of the construction is to make the optical path traverse all cells following $(1,1)\to(1,2)\to(2,2)\to(2,1)$, and then reverse this trajectory going back to $(1,1)$. We can generalize this construction to any $N\times M$ square mesh. Roughly speaking, we make the optical path traverse the first row from left to right (e.g., $(1,1)\to(1,2)\to\cdots\to(1,M)$). Then by setting $\text{TBU}_{1M-2M}$ to cross state, we direct the optical path down to the next row, cell $(2,M)$. Next, the optical path will follow $(2,M)\to(2,M-1)\to\cdots\to(2,1)$, and so on and so forth. With a proper setting of all TBUs, the optical path will follow a zigzag path traversing to the last cell in the $N$-th row, and then reversely go back to $(1,1)$.

Formally, we provide the construction in a general $N\times M$ square mesh. We set one peripheral TBU, $\text{TBU}_{10-11}$, all non-peripheral vertical TBUs, and $(N-1)$ non-peripheral horizontal TBUs,
\begin{equation*}
\begin{aligned}
    \{\text{TBU}_{1M-2M},\,\text{TBU}_{21-31},\,\text{TBU}_{3M-4M},\,\cdots
    \\ 
    ,\,\underbrace{\text{TBU}_{(N-1)M-NM} \text{ or } \text{TBU}_{(N-1)1-N1}\}}_{\text{depends on if $N$ is odd or even}}
\end{aligned}
\end{equation*}
to cross state, while all other TBUs are set to bar state. Then an undirected optical path with start and end at the floating nodes of $\text{TBU}_{10-11}$ has length of $(4NM+1)$.
\end{proof}

Theorem~\ref{thm:maximum_length} implies that if $x$ is not an integer in the range $[1,\, 4NM+1]$, then it will not be realizable in an $N\times M$ square mesh. Then is any integer path length in the range $[1,\, 4NM+1]$ realizable? Unfortunately, the answer is no, as the curious reader might have found for themselves in trying to synthesize a path of length $3$ or $7$ in a $2\times 2$ square mesh. At the end of this section, we will determine the realizability condition, but before that, we need some further understandings on path length. 

Consider an undirected optical path with start point on the left side of a $2\times 2$ mesh shown in Fig.~\ref{fig:proof}~(d). Based on where the end point is, there are three types: (i)~{Type S}: the gray path, where the end point is located on the {\em same} side as the start point; (ii)~{Type A}: the red path, where the end point is located on the {\em adjacent} side (i.e., top or bottom in this case) of the start point; and (iii)~{Type O}: the pink path, where the end point is located on the {\em opposite} side (i.e., the right side in this case) of the start point. We observe that the path length possesses very different characteristics for these three cases, as summarized in the following Theorem~\ref{thm:length_depends_case}.

\begin{figure}[!htb]
    \centering
    \includegraphics[width=0.8\linewidth]{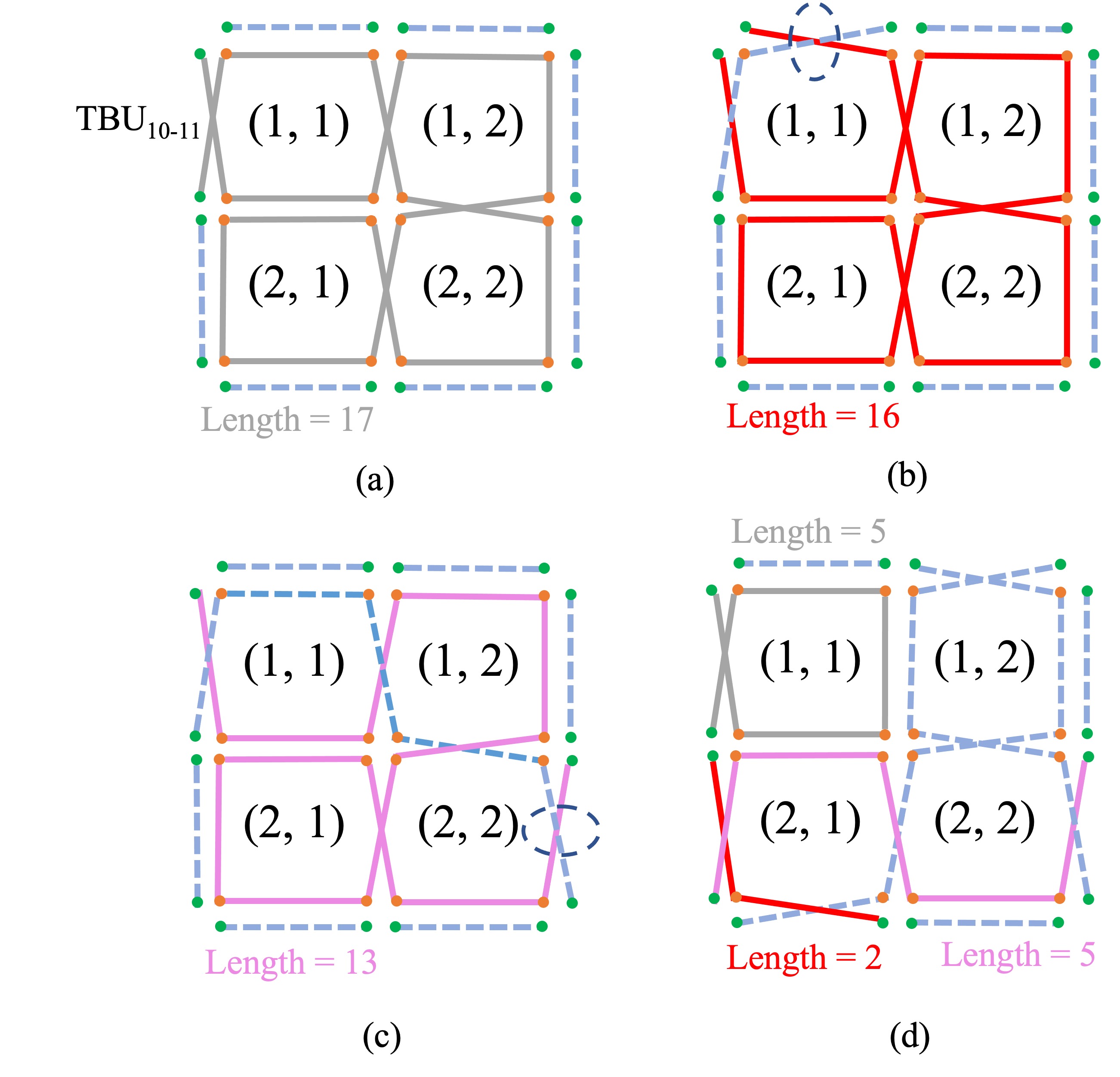}
    \caption{(a)~Construction for maximum length $(4MN+1)$ in an $N\times M$ square mesh demonstrated in the case of $M=N=2$. (b)~Based on (a), we change the top left peripheral horizontal $\text{TBU}_{01-11}$ to cross state, yielding a path of length $4MN$. (c)~Based on (a), we change the second top right peripheral vertical $\text{TBU}_{2M-2(M+1)}$ to cross state, yielding a path of length $(4MN+1-2M)$. (d)~We fix the start point of an undirected optical path to the left side; based on the end point location, there are three different path types.}
    \label{fig:proof}
\end{figure}

\begin{theorem}
\label{thm:length_depends_case}
Consider an undirected optical path in an $N\times M$ square mesh. Without loss of generality, we assume its start point is located on the left side of the square mesh, and denote its path length by $l$.
\begin{enumerate}[label = (\arabic*)]
    \item Type S: If the path's start and end nodes are located on the \underline{same} side, then $l \equiv 1 \pmod{4}$.
    \item Type A: If the path's start and end nodes are located on \underline{adjacent} sides, then $l \equiv 0 \text{ or } 2 \pmod{4}$.
    \item Type O: If the path's start and end nodes are on \underline{opposite} sides, then $l \equiv 3 \pmod{4}$ if $M$ is odd, and $l \equiv 1 \pmod{4}$ if $M$ is even.
\end{enumerate}
Here $l \equiv d \pmod{4}$ (where $d=0,1,2,3$) means $l$ has a remainder of $d$ when divided by $4$.
\end{theorem}

\begin{proof}
Theorem~\ref{thm:oobvious}~(6) states that any undirected optical path can be expressed using the notation `$\cdots$ VHVH $\cdots$', where `V' and `H' stand for vertical and horizontal TBUs, respectively, starting from the start node. We assume the start point is located on the left side; thus, the first TBU seen by this path must be a vertical one (see Fig.~\ref{fig:undirected_simple_graph_representation}~(c)). 

Now if a path belongs to type A, then the last TBU seen by this path must be horizontal (see Fig.~\ref{fig:undirected_simple_graph_representation}~(c)). Thus, in this case, the path follows `VHVH$\cdots$H', which implies that this path must see $2d$ TBUs (i.e., $d$ vertical and $d$ horizontal) in total. Namely, $l$ is even. For later consistency, we write the path length in this case as $l \equiv 0 \text{ or } 2 \pmod{4}$.

If a path belongs to type S, then the last TBU seen by this path must be vertical, which indicates that this path follows `VHVH$\cdots$V'. We denote the total number of `H's on this path by $d$. Then, the path length is $l=2d+1$ since the number of `V's is one larger than that of `H's. Now, if we can prove $d$ is even, then we will attain the desired conclusion $l \equiv 1 \pmod{4}$ for type S. Let us complete the proof by assuming there are $d_0$ `H's corresponding to the path going in the forward direction (from left to right). Then, since both the start and end points are on the left side, there must be $d_0$ `H's corresponding to the backward direction (from right to left) as well. Otherwise, the end point cannot be on the left side. Thus, the total number of `H's $d=2d_0$, implying $d$ is indeed even.

Now we deal with type O. The proof for this case is similar to that for type S. If a path belongs to type S, then the last TBU seen by this path must be vertical, and this path follows `VHVH$\cdots$V'. If there are $d_0$ `H's corresponding to going in the backward direction (from right to left), then there must be $(d_0+M)$ `H's corresponding to going in the forward direction, since the end point is located at the right side. Thus, there are $(2d_0+M)$ `H's in total. The number of `V's is $(2d_0+M+1)$, and the path length $l=\#H+\#V=4d_0+2M+1$. This implies that $l \equiv 3 \pmod{4}$ if $M$ is odd, and $l \equiv 1 \pmod{4}$ if $M$ is even.
\end{proof}

\begin{corollary}\label{cor:length_different_type}
Following the same assumptions of Theorem~\ref{thm:length_depends_case}, the maximum and minimum path length for different types is listed below, and the bound is achievable by some circuit configuration.
\begin{enumerate}[label = (\arabic*)]
    \item Type S: $1 \leq l\leq 4MN+1$
    \item Type A: $2\leq l \leq 4MN$
    \item Type O: $2M+1 \leq l \leq 4MN - 2M + 1$
\end{enumerate}
\end{corollary}

\begin{proof}
For (1), if we set a peripheral vertical TBU at the left side to bar state, then its two floating nodes build up an undirected optical path of length $1$. For the upper bound, we have already provided a construction for the maximum path length $(4MN+1)$ in Theorem~\ref{thm:maximum_length} and the constructed path belongs to type S.

For (2), the minimum path length $2$ is demonstrated by the red path in Fig.~\ref{fig:proof}~(d). For the upper bound, first because $l$ cannot be $(4MN+2)$ or larger due to Theorem~\ref{thm:maximum_length} and $l$ must have remainder of $0$ or $2$ due to Theorem~\ref{thm:length_depends_case}, $l$ has to be smaller than $4MN$. Moreover, we can provide a circuit configuration to achieve the length $4MN$ based on the construction shown in Theorem~\ref{thm:maximum_length}: we further change $\text{TBU}_{{01}-{11}}$ to cross state, yielding a path of length $4MN$ belonging to type A. See Fig.~\ref{fig:proof}~(b) for an illustration in a $2\times 2$ square mesh.

For (3), that the minimum path length equals $(2M+1)$ is straightforward, as demonstrated by the pink path in Fig.~\ref{fig:proof}~(d). Namely, the minimum path is attained when the pink path attempts to directly go from left to right. However, due to the special topology, it will pass through $M$ horizontal TBUs, and $(M+1)$ vertical TBUs in turns, leading to $l=2M+1$. For the upper bound, we first notice that a path of type O consumes at most one edge of each peripheral horizontal TBU (see Fig.~\ref{fig:proof}(c)). Since there are $2M$ peripheral horizontal TBUs, and $M(N-1)$ non-peripheral horizontal TBUs, a path of type O represented by `VHVH$\cdots$HV' at most consumes $2NM$ `H's because of: 
$$
\underbrace{2M\times 1}_{\text{`H' by peripheral horizontal TBUs}} + \underbrace{M(N-1) \times 2}_{\text{`H' by non-peripheral horizontal TBUs}}
$$
However, $2NM$ `H's will make the path's start and end point both at the left side, and the path will be of type S not type O. To enforce the path being type O (i.e., end at the right side), we have to subtract $2NM$ by $M$ at least. In summary, the maximum number of `H's a path of type O can consume is $M(2N-1)$, where $NM$ 'H's correspond to going in the forward direction (i.e., from left to right), and the remaining for the backward direction. As we have explained in Theorem~\ref{thm:maximum_length}, the number of `V's is one larger than the number `H's. Thus, we have $l\leq 2M(2N-1)+1=4MN-2M+1$. Now, we demonstrate $(4MN-2M+1)$ is achievable. Similarly, based on the construction shown in Theorem~\ref{thm:maximum_length}, we make one modification: we further change $\text{TBU}_{{2M}-{2(M+1)}}$ to cross state, yielding a path of length $(4MN+1-2M)$ belonging to type O. See Fig.~\ref{fig:proof} (c) for an illustration in a $2\times 2$ square mesh.
\end{proof}

Theorem~\ref{thm:maximum_length} and Corollary~\ref{cor:length_different_type} provide us the information on path length based on the path type. When presenting Theorem~\ref{thm:maximum_length} and Corollary~\ref{cor:length_different_type}, we assume that the start point is located on the left side, while it should be straightforward to generalize them to other cases, such as the start point on the right/bottom/top side. We do emphasize that for path type O, the generalization should be done carefully, as suggested by the following Corollary~\ref{cor:length_typeO_start_at_top}. 

\begin{corollary}\label{cor:length_typeO_start_at_top}
Consider an undirected optical path in an $N\times M$ square mesh. We assume its start point is located on the \underline{top side} of the square mesh, and denote its path length by $l$. If it belongs to type O (i.e., start and end point at opposite sides), then $l \equiv 3 \pmod{4}$ if $N$ is odd, and $l \equiv 1 \pmod{4}$ if $N$ is even. Moreover, $2N+1 \leq l \leq 4MN - 2N + 1$.
\end{corollary}

In essence, for type O, when the start point is on the left or right side, the condition should be depicted using the number of columns $M$; and when the start point is on the top or bottom side, the condition should be depicted using the number of rows $N$. Now, we are ready to present our first main theorem about the realizability of a single path.

\begin{theorem}\label{thm:realize_single_path}
\textbf{\textup{(Main Result I)}} For an $N\times M$ square mesh and a desired path length $x$, we denote three integer sets:
$$
\begin{aligned}
\Gamma_\star &= \{d\,|\,d \equiv 0,1,2 \pmod{4}, \; 1 \leq  d \leq 4MN+1\}\\
\Gamma_M &= \{d\,|\,d \equiv 3 \pmod{4}, \; 2M+1 \leq   d  \leq 4MN+1-2M\}\\
\Gamma_N &= \{d\,|\,d \equiv 3 \pmod{4}, \; 2N+1 \leq   d  \leq 4MN+1-2N\}\\
\end{aligned}
$$
\begin{enumerate}[label = (\arabic*)]
    \item If both $N$ and $M$ are even, then any $x\in\Gamma_\star$ is realizable; $x\not\in\Gamma_\star$ is not.
    \item If $N$ is even and $M$ is odd, then any $x\in\Gamma_\star\cup\Gamma_M$ is realizable; $x\not\in\Gamma_\star\cup\Gamma_M$ is not.
    \item If $N$ is odd and $M$ is even, then any $x\in\Gamma_\star\cup\Gamma_N$ is realizable; $x\not\in\Gamma_\star\cup\Gamma_N$ is not.
    \item If both $N$ and $M$ are odd, then any $x\in\Gamma_\star\cup\Gamma_N\cup\Gamma_M$ is realizable; $x\not\in\Gamma_\star\cup\Gamma_N\cup\Gamma_M$ is not.
\end{enumerate}
\end{theorem}

\begin{proof}
For (1), using Theorem~\ref{thm:length_depends_case} with Corollary~\ref{cor:length_different_type} and~\ref{cor:length_typeO_start_at_top}, we readily obtain for any $x\not\in\Gamma_\star$ that it is not realizable. Similar conclusions hold true for cases~(2)-(4). The remaining task is to provide constructions showing that for any $x$ in our defined set, it is indeed realizable.

In the following, we will demonstrate the construction method using a small square mesh; extending to a general $N\times M$ mesh is straightforward. To begin, we demonstrate the construction method for case~(1) using a $2\times 2$ square mesh in Fig.~\ref{fig:construction}. The key is to use all cells in a zigzag order: $\{(1,1), (1,2),\cdots, (1,M), (2,M), (2,M-1),\cdots,(2,1),(3,1),(3,2),\cdots\}$. Take $l\equiv 1 \pmod{4}$ as an example. The set $\{5,9,\cdots,4MN+1\}$ contains $NM$ integers, where $9$ occurs in the second place, and thus we will use cell $(1,1)$ and $(1,2)$ as shown in the second sub-figure in the top row of Fig.~\ref{fig:construction}. As another example, the set $\{2,6,10,\cdots, 4MN-2\}$ contains $NM$ integers, where $6$ occurs as the second, and thus we will use cell $(1,1)$ and $(1,2)$ as shown in the second sub-figure in the middle row of Fig.~\ref{fig:construction}.

\begin{figure*}[!tb]
    \centering
    \includegraphics[width=0.9\textwidth]{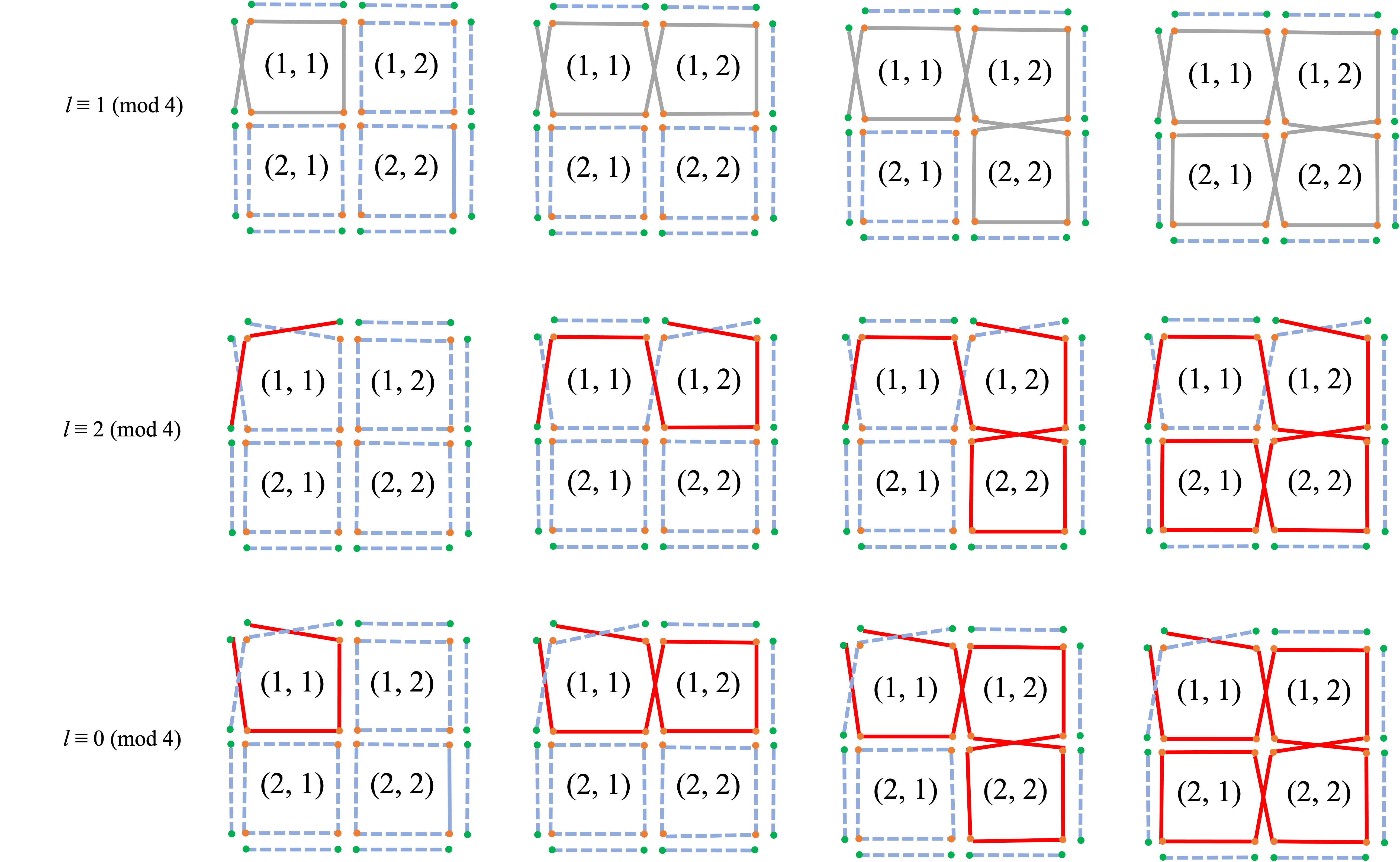}
    \caption{Illustration of construction for $\Gamma_\star$ using a $2\times 2$ square mesh (i.e., $N=M=2$). Note that $l=1$ is trivial, and is not shown in the first row.}
    \label{fig:construction}
\end{figure*}

For case~(2), we can apply the construction method we show for case~(1) to deal with $x\in\Gamma_\star$, and we only need to provide a construction method for those $x\in\Gamma_M$. We demonstrate our construction in Fig.~\ref{fig:construction2}. The set $\Gamma_M=\{2M+1, 2M+5,\cdots, 4NM-2M+1\}$ contains $(NM-M+1)$ integers. To realize the first desired path length $(2M+1)$, we will use all cells in the first row (i.e., cell $\{(1,1),(1,2),\cdots,(1,M)\}$), as shown in the first sub-figure in Fig.~\ref{fig:construction2}. Then for the remaining other desired path length, we will gradually exploit one additional cell in the zigzag order: $\{(2,M),(2,M-1),\cdots,(2,1),(3,1),\cdots\}$. Finally, the construction for cases~(3) and (4) should already be understood as they are similar to case~(2).

\begin{figure}[!htb]
    \centering
    \includegraphics[width=0.9\linewidth]{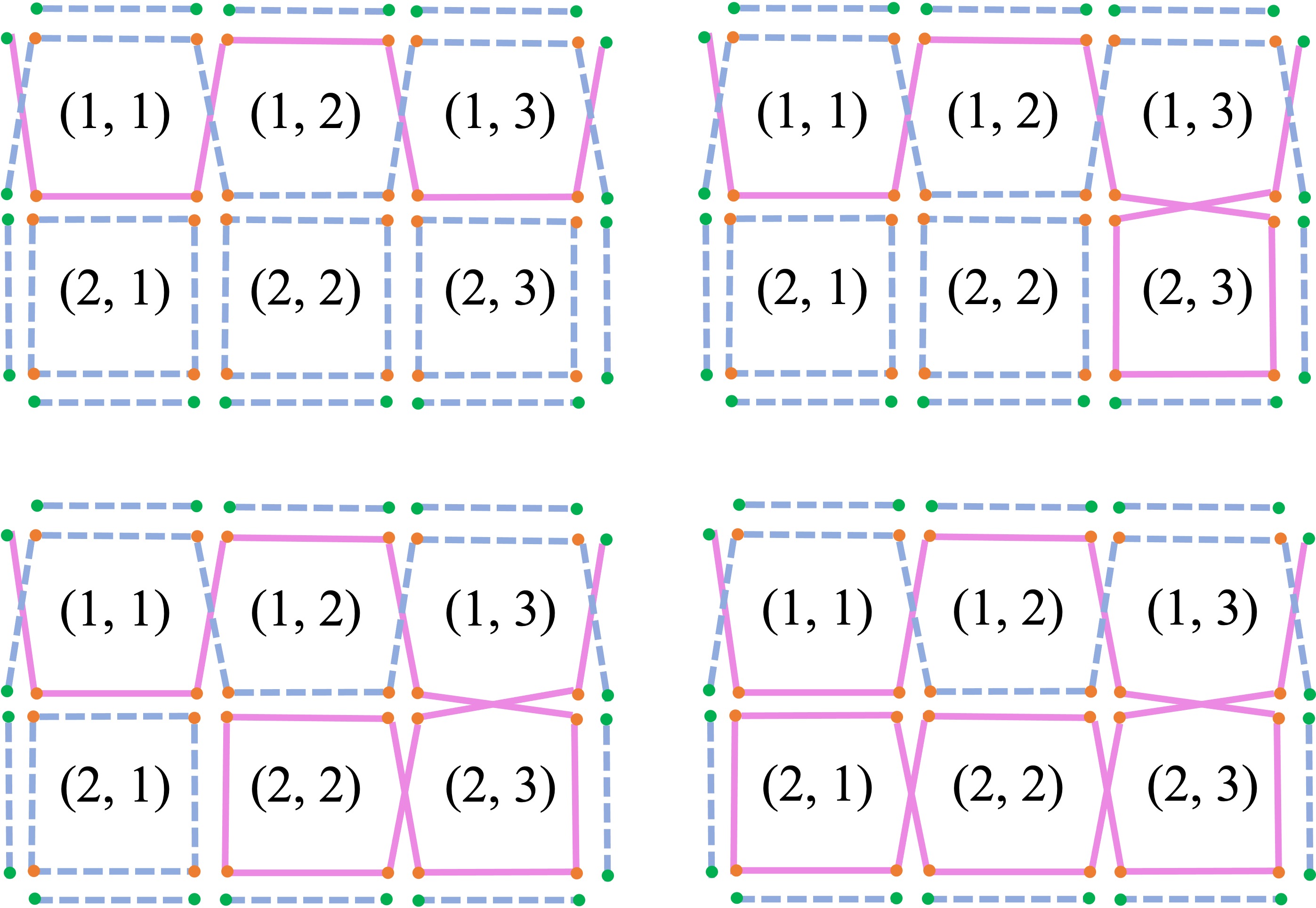}
    \caption{Illustration of construction for $\Gamma_M$ using a $2\times 3$ square mesh (i.e., $N=2$ and $M=3$).}
    \label{fig:construction2}
\end{figure}

\end{proof}

Thus far,  we have thoroughly answered the question of what path length $x$ is realizable in an $N\times M$ square mesh. Our findings provide valuable guidance when routing optical signals, such as where to put input and output nodes. We will consider such applications in Section~\ref{sec:applications}.


\subsection{Realizability of Multiple Paths}\label{sec:multi_path}

Building upon the single-path case, a more interesting question is: \textbf{In an $N\times M$ square mesh, can we find a circuit configuration to realize $y$ paths, each of length $x$?}. In this section, we make the number of paths a variable $y$ (instead of fixing it to $1$ as the previous section did), and investigate the maximum value of $y$ given the value of $x$. As a quick example, when $x=1$, we know the maximum value of $y$ is $(2N+2M)$ because path $1$ is only realizable using peripheral TBUs.


\begin{theorem}\label{thm:multi_path_average}
For a fixed circuit configuration of an $N\times M$ square mesh, we collectively denote the lengths of all $(2N+2M)$ undirected optical paths by $\Gamma=[l_1,l_2,\cdots,l_{2N+2M}]$. Then the sum of all undirected optical paths $\sum_{i=1}^{2N+2M} l_i$ can be written in the format $(2N+2M+4k)$, for some $k\in\{0,1,\cdots,NM\}$. Moreover, when the path sum equals $(2N+2M+4k_0)$, then 
\begin{enumerate}[label = (\arabic*)]
    \item The path average $\Bar{\Gamma}=1+\frac{2k_0}{N+M}$.
    \item For any undirected optical path, its path length $1 \leq l_i \leq 4k_0+1$.
    \item The path variance $\sigma^2(\Gamma)\leq\frac{8k_0^2}{N+M}-\frac{4k_0^2}{(N+M)^2}$.
\end{enumerate}
and the bound in (2)-(3) are achievable.
\end{theorem}

\begin{proof}
We first prove that the path sum is in the format $(2N+2M+4k)$, for some $k\in\{0,1,\cdots,NM\}$. According to Theorem~\ref{thm:oobvious}~(1), there are $(2NM+M+N)$ TBUs in total. Since each TBU has two edges that can be used by the optical signal (regardless of bar or cross state), there are $(4NM+2N+2M)$ edges altogether. Furthermore, each edge occurs at most in one undirected optical path, so the maximum path sum equals $(4NM+2N+2M)$. We should notice that achieving length sum $(4NM+2N+2M)$ is an ideal case in which all of the edges of the square cell are used and that there is no untraveled closed loop in the circuit. Next, we consider the case when closed loops exist; we will show that each closed loop not being used by any undirected optical path must have length $4k$, which further implies that the path sum equals $(4NM+2N+2M - 4k)$, where $k=0,1,\cdots, NM$. Note that the construction for each $k\in\{0,1,\cdots,NM\}$ is already shown in the first row of Fig.~\ref{fig:construction}. The proof is straightforward if we notice that a closed loop indicates that the loop passes an even number of 'V's and even number of 'H's, so that it can start from a non-floating node and also ends at the same node. Thus, its length is $4k$ because the path length equals the sum of the number of 'V's and 'H's.


Now, we prove the three statements. The first statement about path average is trivial, because the path sum equals $(2N+2M+4k_0)$ while the number of paths is $(2N+2M)$. 

For the second statement, we first emphasize that Theorem~\ref{thm:maximum_length} is a special case of Theorem~\ref{thm:multi_path_average}~(2) when $k_0=NM$. Now, we notice that when the path sum equals $(2N+2M+4k_0)$, there will be $k_0$ square cells consumed by some undirected optical paths. If the maximum path length is larger than $(4k_0+1)$, following the same thought as in Theorem~\ref{thm:maximum_length}, we see that it will lead to at least $3$ non-floating nodes, or $(k_0+1)$ square cells used. Either case will result in a contradiction.

To prove the third statement, we retrieve a conclusion from statistics: If a random variable lies in the range $[a,b]$, and its mean is $\mu$, then its variance is upper bounded by $(\mu-a)(b-\mu)$. In our case,  substituting $a$ with $1$, $b$ with $(4k_0+1)$, and $\mu$ with $(1+2k_0/(N+M))$, proves statement~(3).

Last, but not least, we provide a construction achieving the bound shown in statements~(2) and~(3). We enforce one path length to be $(4k_0+1)$ and the remaining $(2N+2M-1)$ path lengths all to be $1$. The construction for such is already shown in the first row of Fig.~\ref{fig:construction}. Now, the variance is given by:
\begin{equation*}
\begin{aligned}
    \sigma^2(\Gamma)&=\frac{1}{2N+2M}\sum_{i=1}^{2N+2M}(l_i-\bar{\Gamma})^2\\
    &=\frac{(4k_0+1-\bar{\Gamma})^2 + (2N+2M-1)(1-\bar{\Gamma})^2}{2N+2M}
\end{aligned}
\end{equation*}
Further substituting the path mean $\bar{\Gamma}=1+2k_0/(N+M)$ as provided in Theorem~\ref{thm:multi_path_average} (1) into the above expression, we have:
\begin{equation*}
\begin{aligned}
    \sigma^2(\Gamma)&=\frac{1}{2N+2M}(4k_0-\frac{2k_0}{N+M})^2+(1-\frac{1}{2N+2M})\frac{4k_0^2}{(N+M)^2}\\
    &=\frac{8k_0^2}{N+M}-\frac{4k_0^2}{(N+M)^2}\\
\end{aligned}
\end{equation*}
which is exactly the expression shown in Theorem~\ref{thm:multi_path_average} (3).

\end{proof}

The following theorem quantifies the multi-path routing ability of an $N\times M$ square mesh.

\begin{theorem}(\textbf{Main Result II})\label{thm:upper_bound_multi_path}
If $y$ paths each of length $x$ can be realized in an $N\times M$ square mesh with a circuit configuration, then we have:
$$
y \leq \min \{\lfloor\frac{4NM}{x-1}\rfloor, 2M+2N, C_1, C_2\}
$$
where $\lfloor\cdot\rfloor$ represents rounding to the integer below. $C_1$ represents an additional constraint active under the case when $x$ is even:
\begin{equation*}
    \begin{aligned}
    C_1 = \left\{\begin{array}{cc}
         4, &\text{If $x$ is even}\\
         +\infty, &\text{Else}\\
    \end{array}\right.
    \end{aligned}
\end{equation*}
Similarly, $C_2$ is an extra constraint active under the case when $x\equiv 3 \pmod{4}$:
\begin{equation*}
    \begin{aligned}
    C_2=\left\{\begin{array}{cc}
          0, &\text{If $x\equiv 3 \pmod{4}$, $N$ and $M$ are even}\\
          0, &\text{If $x\equiv 3 \pmod{4}$, $N$ is even, $M$ is odd, $x\not\in\Gamma_M$} \\
          0, &\text{If $x\equiv 3 \pmod{4}$, $N$ is odd, $M$ is even, $x\not\in\Gamma_N$} \\
          0, &\text{If $x\equiv 3 \pmod{4}$, $N$ and $M$ are odd, $x\not\in\Gamma_N\cup\Gamma_M$} \\
          +\infty, &\text{Else}\\
    \end{array}\right.
    \end{aligned}
\end{equation*}
\end{theorem}
\begin{proof}
The bound $(2M+2N)$ is trivial; since we have $(2M+2N)$ undirected optical paths in total, $y$ cannot be larger than this bound. The bound $C_2$ reuses the conclusion of Theorem~\ref{thm:realize_single_path}, imposing an additional requirement when $x\equiv 3 \pmod{4}$. Additionally, we have the following inequality because the path sum cannot be larger than $(2N+2M+4NM)$:

\begin{equation*}
    xy+\bar{a}(2N+2M-y)\leq 2N+2M+4MN
\end{equation*}
where $\bar{a}$ represents the average path length of the remaining $(2N+2M-y)$ paths. Noticing that $\bar{a}\geq 1$ and $2N+2M-y\geq 0$, we have:
\begin{equation*}
    xy+1(2N+2M-y)\leq 2N+2M+4MN \; \rightarrow \; y\leq\frac{4NM}{x-1} .
\end{equation*}
Our remaining task is to justify the bound $C_1$. In essence, the bound $C_1$ states that a square mesh at most realizes four paths of a specified length $x$ if $x$ is even. To prove this, we recall that in Theorem~\ref{thm:length_depends_case}, an even-length path is only possible if the path belongs to type A. We consider the case when $M$ and $N$ are both sufficiently large compared to $x$ (e.g., $M, N \geq 2x$). As we are merely looking for the upper bound of $y$, imposing this additional assumption is not restrictive; however, it also implies that the bound $C_1=4$ might not be tight for the case when $M$ and $N$ are comparable with $x$ (e.g., $x=4$, $N=1$, and $M=1$). Under this assumption, if there are $d_0$ paths of length $x$ realizable at the top left corner (i.e., with start node at the left side, and end node at the top side) of the square mesh, then by symmetry, there will be $4d_0$ paths of length $x$ realizable in total since there are four corners (see Fig.~\ref{fig:even_path}). In the following, we will prove that the maximum value of $d_0$ equals $1$ in the case $x=4$, justifying the bound $C_1$. Generalizing the proof to an arbitrary even $x$ is straightforward, and is omitted here.

Fist, as shown in Fig.~\ref{fig:even_path}, we prove that for all paths of length four realizable at the top-left corner using some circuit configurations, they all must pass through node $A$. As we already mention in the proof of Theorem~\ref{thm:length_depends_case}, a path of length four will follow 'VHVH'. We consider where the optical path is after passing through the first 'VH' pair, and is about to pass through the second 'VH' pair. For simplicity, this location is referred to as the intermediate node. Since the optical path starts from the left side, the intermediate node must be located at a node in the red circle shown in Fig.~\ref{fig:even_path} after passing through the first 'VH' pair. Furthermore, since the optical path will pass through the second pair 'VH' and then reach the top side, we reverse this process and find that the intermediate node must be located in the pink circle. Thus, the intersection of the red and pink circle uniquely determines that a path of length four at the top-left corner must pass through node~A.

Next, we state an obvious conclusion: For a fixed circuit configuration, if we know that two undirected optical paths pass through the same node, then these two paths actually are identical (i.e., they refer to the same path). With the above two pieces, we can complete our proof. Given a circuit configuration, we assume that there are two different paths of length four synthesizable at the top-left corner. Using the first conclusion, we know both of them pass through node A. Then, using the second conclusion, we see that these two paths are actually identical.

\begin{figure}[!htb]
    \centering
    \includegraphics[width=0.8\linewidth]{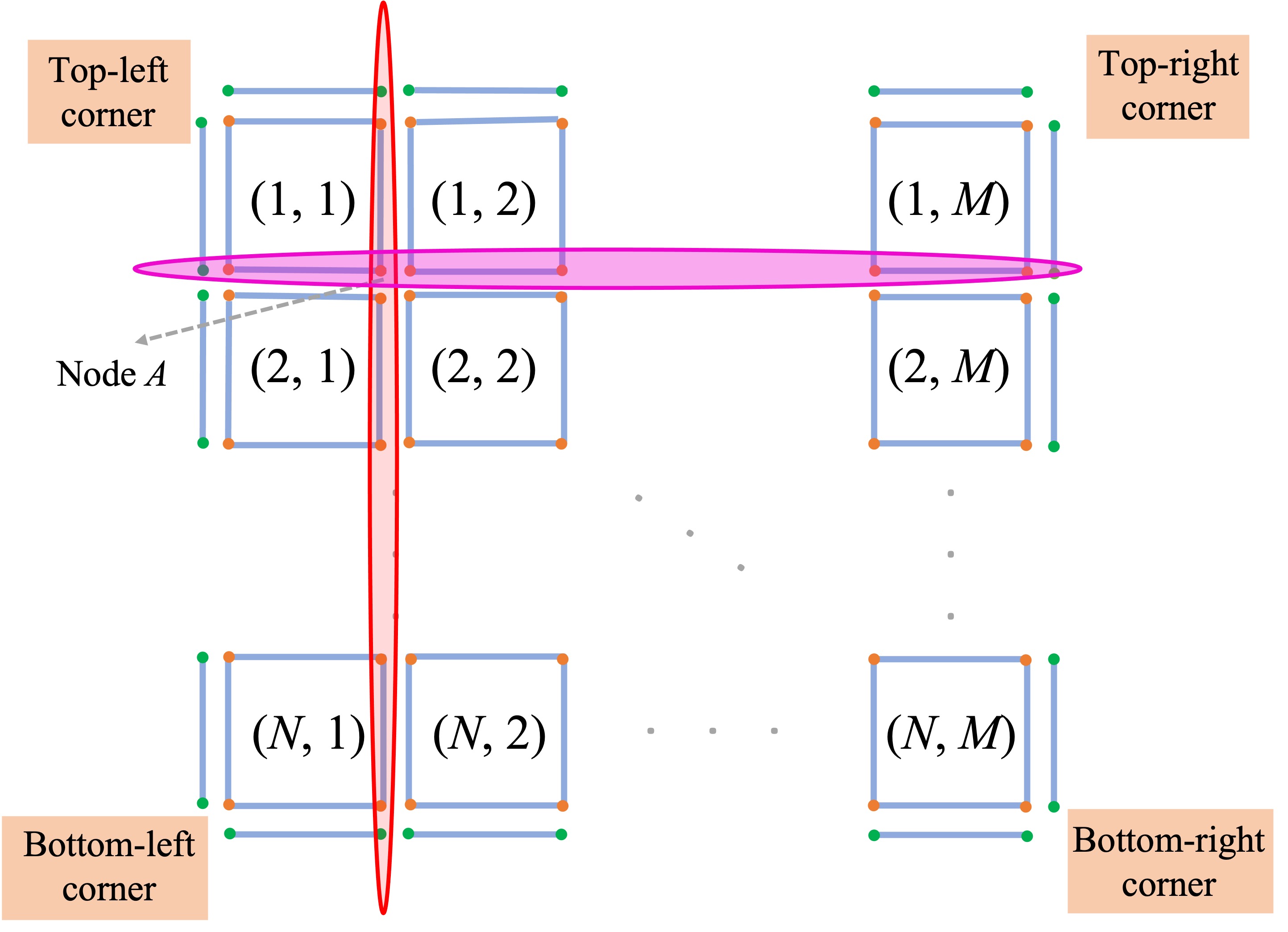}
    \caption{For all paths of length four realizable at the top-left corner using some circuit configurations, they all must pass through node $A$, which is the intersection of the red and pink circle.}
    \label{fig:even_path}
\end{figure}
\end{proof}

In Fig.~\ref{fig:visualization_theorem_y_upper_bound}~(a), we use a $2\times 3$ square mesh as an example to demonstrate how tight our provided upper bound in Theorem~\ref{thm:upper_bound_multi_path} is. Note that in a $2\times 3$ square mesh, there are already $17$ TBUs, resulting in $2^{17}$ circuit configurations. Brute-force enumeration of all configurations is only barely time affordable.

\begin{figure}[!htb]
    \centering
    \includegraphics[width=1.0\linewidth]{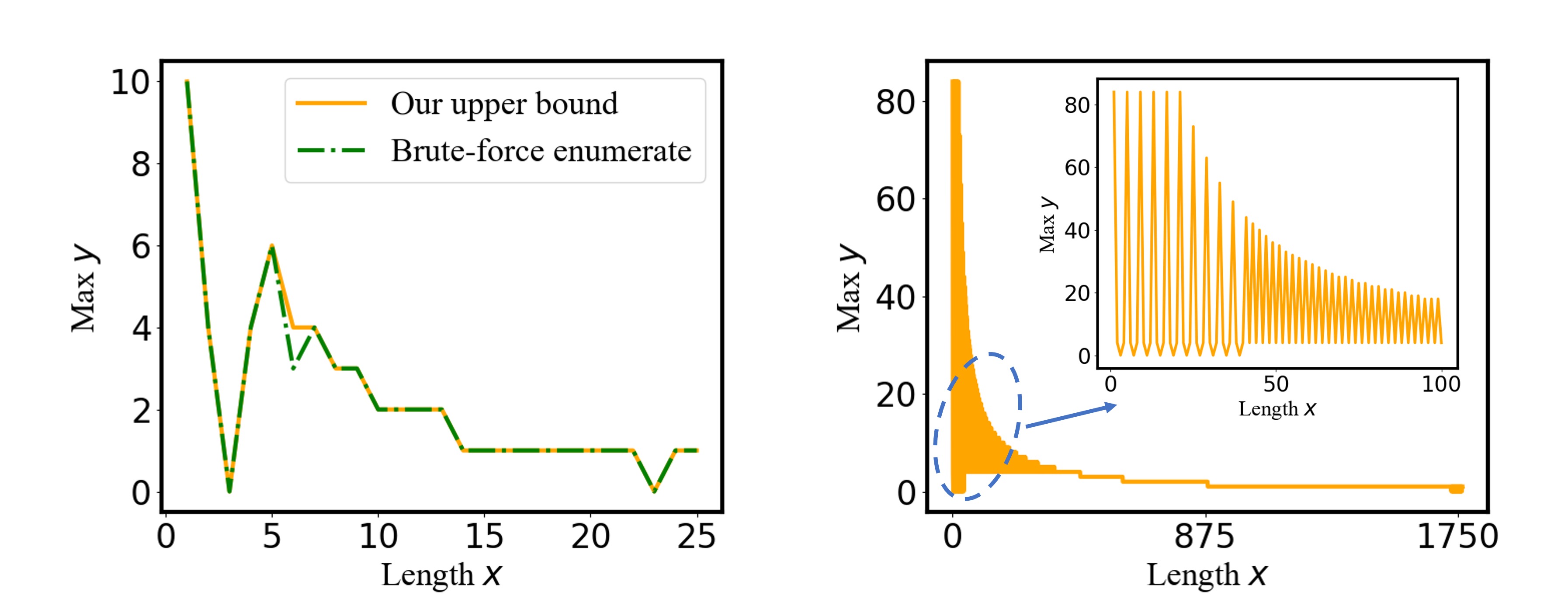}
    \caption{(a)~In a $2\times 3$ square mesh, comparison of our upper bound provided in Theorem~\ref{thm:upper_bound_multi_path} and the real maximum value of $y$ by brute-force enumerating all circuit configurations. (b)~Plot of our estimated upper bound for $y$ versus different $x$ in a $21\times 21$ square mesh.}
    \label{fig:visualization_theorem_y_upper_bound}
\end{figure}

Extending the aforementioned theorems to other topologies, such as triangular and hexagonal meshes, is straightforward and is summarized in the Appendix. An important observation from such analysis is that square meshes are not able to implement some lengths in the range $[1,4NM]$, while triangular and hexagonal meshes do not have this problem. In the next section, we will discuss some potential implications and usage of our findings.

\section{Implications and Applications}\label{sec:applications}

In this section, we present several potential applications based on our findings. Note that the PPIC we present here is of relatively small size consistent with demonstrations to date in the literature
(e.g., $1\times 2$ square mesh shown in Fig. 3 of~\cite{Zhuang15programmable}). Our described applications become even more appealing as PPIC sizes scale up in the future.

\subsection{Guidelines for Setting $N$ and $M$}

In this subsection, we demonstrate how to use our derived conclusions to determine the appropriate values of $N$ and $M$ given a collection of integer path lengths $\Lambda$ that we wish to realize on a square mesh. Note that this question form has real-world implications. For instance, when we want to synthesize a fourth order finite impulse filter with a square mesh as the delay line, then $\Lambda$ should be an arithmetic sequence with four elements (e.g., $[2,4,6,8]$). As another example, if we want to use the square mesh to route four optical signals while keeping their phases equal, then $\Lambda$ should be a collection containing four identical integers (e.g., $[3,3,3,3]$).

Consider the first example where we want to realize path length collection $\Lambda_1=[2,4,6,8]$ on a square mesh. Suppose we want the row and column to be balanced (i.e., $N=M=x\in\mathbb{Z}^{+}$), and we try to determine $x$ using our previous conclusions. Theorem~\ref{thm:maximum_length} imposes the first inequality constraint: $8\leq 4NM+1=4x^2+1$, which gives $x\geq 2$. Then, we use the path sum conclusion in Theorem~\ref{thm:multi_path_average}, leading to another inequality constraint: $2+4+6+8< 2N+2M+4NM=4x+4x^2$, which gives $x\geq 2$. As readers can explore on their own, a $2\times 2$ square mesh indeed can realize $\Lambda_1=[2,4,6,8]$ simultaneously. Finally, we emphasize that as our theorems provide upper bounds, so using these theorems serve as a necessary condition, but might not be sufficient. Namely, $x<2$ is definite to fail for $\Lambda_1=[2,4,6,8]$, but $x=2$ might not be sufficient either, and $2$ is only a starting search point.

The first example is of relative small scale, and even a direct guess without knowing our conclusions might attain a good result. In the second example, we consider $\Lambda_2 = [6,10,14,18,22,26]$ and balanced row and column (i.e., $N=M=x$). Estimating $x$ without knowing our conclusions is quite difficult. On the other hand, our Theorem~\ref{thm:maximum_length} requires: $26 \leq 4x^2+1$, which gives $x\geq 3$. Moreover, the path sum conclusion in Theorem~\ref{thm:multi_path_average} requires: $6+10+14+18+22+26\leq 2N+2M+4NM = 4x+4x^2$, yielding $x\geq 5$.
Since both of these two constraints must be satisfied, we suggest starting from a $5$-by-$5$ square mesh for this particular $\Lambda_2$. To end this example, notice that a configuration realizing $\Lambda_2 = [6,10,14,18,22,26]$ on an $N$-by-$M$ square mesh with TBU length $L$ and TBU loss $\alpha$ is equivalent to another configuration realizing $\Lambda^\prime = [3,5,7,9,11,13]$ on a square mesh with TBU length $L^\prime=2L$ and TBU loss $\alpha^\prime={\alpha}^2$ from the perspective of magnitude and phase response, as evidenced by Eq.~(\ref{eq:path_length_response}). However, we must emphasize that in terms of synthesizing $\Lambda$ and $\Lambda^\prime$, they are completely different. Specifically, path length $3$ is required in $\Lambda^\prime$, and as stated in Theorem~\ref{thm:realize_single_path}, either $N$ or $M$ must be $1$. This will be more clear in the following third example.

In the third example, we consider $\Lambda_3=[3,5,7,9,11,13]$. Note that in the previous $\Lambda_1$ and $\Lambda_2$, the required path lengths are all even, so that they will not have remainders equal to $3$ when divided by $4$. However, $3$ and $11$ in $\Lambda_3$ both have remainder of $3$ when divided by $4$. As stated in Theorem~\ref{thm:realize_single_path}, at least one of $M$ or $N$ must be odd to handle these cases. Moreover, path length $3$ is even more special; when substituting it as $d$ into the inequality given in Theorem~\ref{thm:realize_single_path}, it requires either $M=1$ or $N=1$. Without loss of generality, we assume $N=1$. As in our previous examples, using Theorem~\ref{thm:maximum_length} and the path sum conclusion, we have: $13\leq 4NM+1$ and $3+5+7+9+11+13\leq 2N+2M+4NM$, yielding $M\geq 8$. As a byproduct, this example suggests that careful treatment needs to be taken if a path length of remainder $3$ when being divided by $4$ is present.

In the fourth example, we consider a reverse example. Can a $2$-by-$2$ square mesh implement $\Lambda_4=[1,18]$? We find $18$ is larger than the maximum allowed path length $4NM+1=17$, so the answer is no. Then, what about $\Lambda_4=[1,2,4,5,8,10]$? We observe that the sum of $\Lambda_4$ is $30$ larger than the allowed $4NM+2N+2M=24$, so the answer is no. As a further follow up, what about $\Lambda_4=[1,1,1,1,2,4,5,10]$? The answer is still no, because now we have eight paths in $\Lambda_4$, reaching the maximum allowable paths $2N+2M=8$ realized in a $2$-by-$2$ square mesh. This motivates us to calculate the mean of $\Lambda_4$, which equals $3.125$ and cannot be written in the form shown in (1) of Theorem~\ref{thm:multi_path_average}. In summary, our theorems provide several criteria to rule out unreasonable path length collections for a predefined size of square mesh with almost no cost. Note that in this reverse example, we try to detect which inequality constraints defined in our theorems $\Lambda_4$ violate. If any of these constraints are violated, we can be certain that $\Lambda_4$ is impossible to realize. However, to fully prove its feasibility when none is violated, we would need to construct a solution.


\subsection{Inverse Measurement and Characterization}

In this subsection, we demonstrate how to use our findings to inversely characterize the value of $\alpha$ once a square-mesh PPIC chip is fabricated. We set all TBUs to cross state, so that the length sum of the $(2N+2M)$ optical paths equals $(2N+2M+4NM)$. We inject $(2N+2M)$ optical sources independently from each of the input nodes, and measure the signals at the corresponding output nodes. Based on the measurement results, we can calculate the ratio of the output signal over the input signal, denoted by $\{r_1,r_2,\cdots,r_{2N+2M}\}$, where each $r_i$ is a complex scalar, including both the information of magnitude and phase response. Based on Eq.~(\ref{eq:path_length_response}), it is straightforward that we have the following relation:

\begin{equation}
    \prod_{i=1}^{2N+2M} |r_i| = \alpha^{\sum_{i=1}l_i}=\alpha^{2N+2M+4MN}
\end{equation}
where $l_i$ represents the length of the $i$-th optical path. It implies that once we have the measurements $\{r_1,r_2,\cdots,r_{2N+2M}\}$, we can inversely estimate $\alpha$ by:

\begin{equation}\label{eq:inverse_alpha}
    \alpha = \exp(\frac{\sum_{i=1}\log |r_i|}{2N+2M+4NM}) .
\end{equation}

Note that in real fabrication, process variation exists, and it is likely each TBU will have a slightly different $\alpha$, denoted by $\{\alpha_1,\alpha_2,\cdots,\alpha_{2N+2M}\}$. Under this circumstance, Eq.~(\ref{eq:inverse_alpha}) actually yields an estimation of their geometric mean, i.e., $\sqrt[2N+2M+4NM]{\alpha_1\alpha_2\cdots\alpha_{2N+2M}}$. As a byproduct, readers might be curious if we can estimate each individual $\{\alpha_1,\alpha_2,\cdots,\alpha_{2N+2M}\}$, not solely their geometric mean. Unfortunately, ambiguity arises when we try to do so as shown in Fig.~\ref{fig:not_unique}.

\begin{figure}[!thb]
    \centering
    \includegraphics[width=0.8\linewidth]{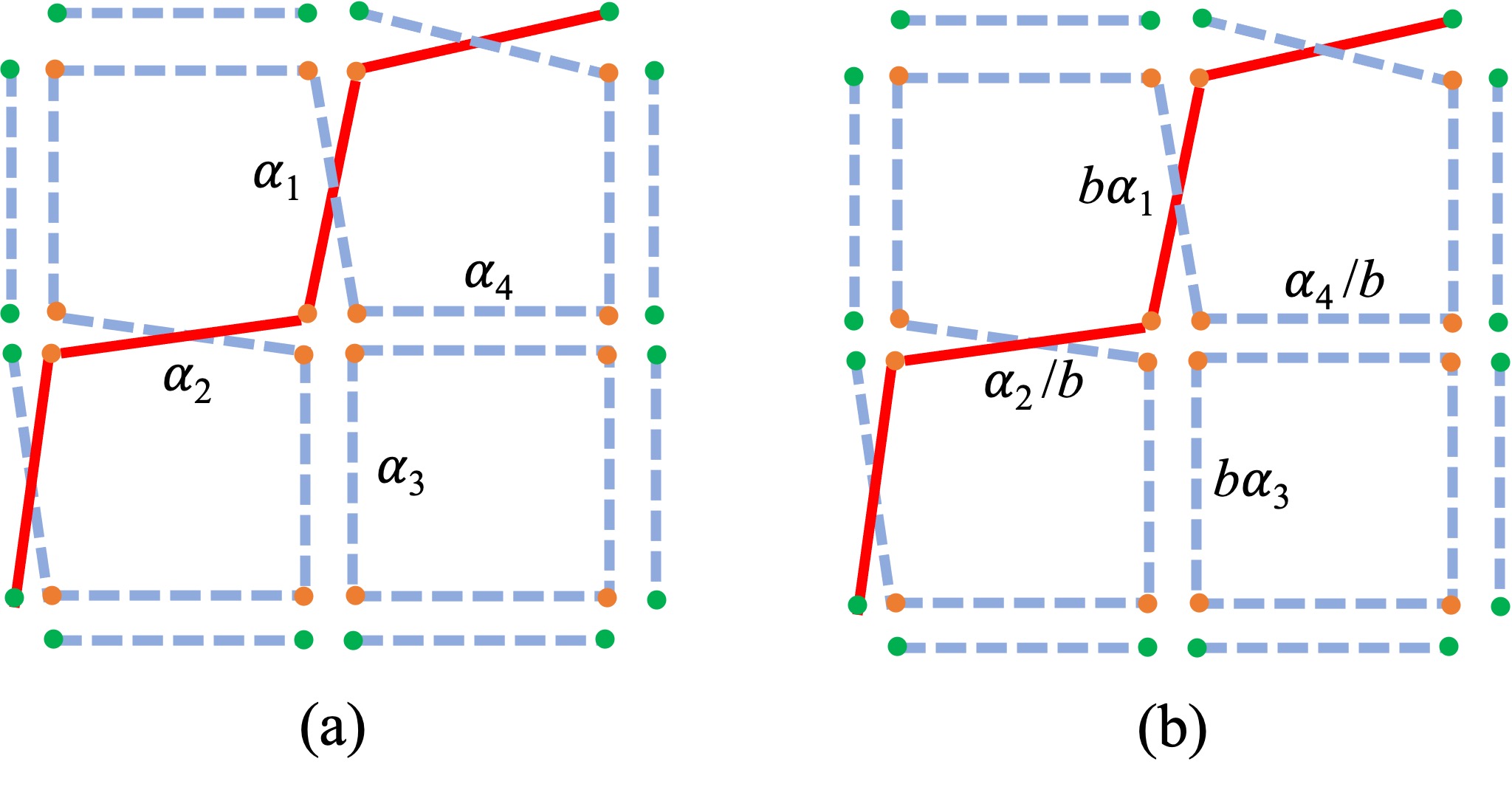}
    \caption{Ambiguity occurs when attempting to determine each individual $\alpha$ value for non-peripheral TBUs. If $\{\alpha_1,\alpha_2,\alpha_3,\alpha_4\}$ is a solution as shown in (a), then (b) will also be a solution. As an example, if $\{\alpha_1,\alpha_2\}$ satisfy the measurement result of the red optical path, then $\{b\alpha_1,\alpha_2/b\}$ will also satisfy the measurement result for some constant $b$.}
    \label{fig:not_unique}
\end{figure}

By analogy to inversely measuring $\alpha$, we can do the same thing for the TBU length $L$. Still setting all TBUs to cross state, then the variable $q$ in Eq.~(\ref{eq:path_length_response}) for any optical path will be zero. Thus, we have the following relation:

\begin{equation}
\begin{aligned}
\sum_{i=1}^{2N+2M} \text{Phase}(r_i) &= -(\sum_{i=1}^{2N+2M} l_i)\cdot\frac{\omega n_\text{eff} L}{c}\\
&=-(2N+2M+4NM)\frac{\omega n_\text{eff} L}{c} 
\end{aligned}
\end{equation}
Note that the equality sign holds under the addition of integer multiples of $2\pi$ (denoted by $2d\pi$ later). If we assume that $n_\text{eff}$ is known, then we can inversely measure the TBU length variable $L$ via:
\begin{equation}\label{eq:estimator_L}
    L=[-\sum_{i=1}^{2N+2M} \text{Phase}(r_i)  + 2d\pi]\frac{c}{(2N+2M+4NM)\omega n_\text{eff} }
\end{equation}
where $d$ is an integer added to make the estimated $L$ meaningful. We note that it is expected that ambiguity of the estimated value of $L$ occurs when doing this inverse characterization due to the periodicity of phase. Similarly, due to process variation, each TBU will have a slightly different length parameter $L$, denoted by $\{L_1,L_2,\cdots,L_{2N+2M}\}$, and Eq.(~\ref{eq:estimator_L}) then estimates their arithmetic mean. 

To end this subsection, we consider a variant of the above approach, focusing on local characterization. We take inverse measurement of $\alpha$ as an example. If we wisely set the TBUs into bar/cross state, we can exclude a few TBUs from being passed though by any of the $(2N+2M)$ optical paths, and the length sum will equal $(2N+2M+4k_0)$, as depicted in Theorem~\ref{thm:multi_path_average}. In this case, the denominator in Eq.~(\ref{eq:inverse_alpha}) should be revised to $(2N+2M+4k_0)$ accordingly, and the calculated $\alpha$ becomes a 'local' estimation. When only part of the square-mesh PPIC is exploited in an application, this local inverse measurement might be more accurate than the previous global one, and of particular interest.

\section{Conclusions}\label{sec:conclusion}

In this paper, we theoretically investigate the routing ability of programmable photonic integrated circuits under the assumption that TBUs are either in cross or bar state. Such an assumption is reasonable to be made in an optical routing application. Based on the compact model of the TBU, we first show that the path length (defined as the number of TBUs that a path passes through) is decisive in the signal response, affecting both phase and magnitude. Next, we provide several theorems rigorously determining what path length is realizable in single-path routing. Then, we approach multi-path routing, providing analytical expressions for the path length sum, and an upper bound on path length variance and the maximum number of realizable paths. Finally, a number of potential optical applications using our observations are illustrated.

\vspace{6pt}
\paragraph{Funding.} Xiangfeng Chen and Wim Bogaerts received funding from the European Research Council through grant 725555 (PhotonicSWARM).

\paragraph{Disclosures.} The authors declare that there are no conflicts of interest related to this article.

\paragraph{Data availability.} Data underlying the results presented in this paper are not publicly available at this time but may be obtained from the authors upon reasonable request.

\bibliography{sample}

\section*{Appendix}\label{sec:appendix}

The conclusions and proofs shown in the main text can be well generalized to other topologies such as hexagonal mesh, triangular mesh, with a few light modifications. In the Appendix, we consider two popular topologies, parallelogram hexagonal mesh, as shown in Fig.~\ref{fig:hexagonal}, and parallelogram triangular mesh, as shown in Fig.~\ref{fig:triangular}. Other topology variants (e.g., concentric hexagonal mesh) won't be covered, while readers can derive themselves by relying on the analysis methods we provide. Theorem~\ref{thm:hexagonal} summarizes the conclusions for a parallelogram hexagonal mesh.

\begin{figure}[!htb]
    \centering
    \includegraphics[width=0.5\linewidth]{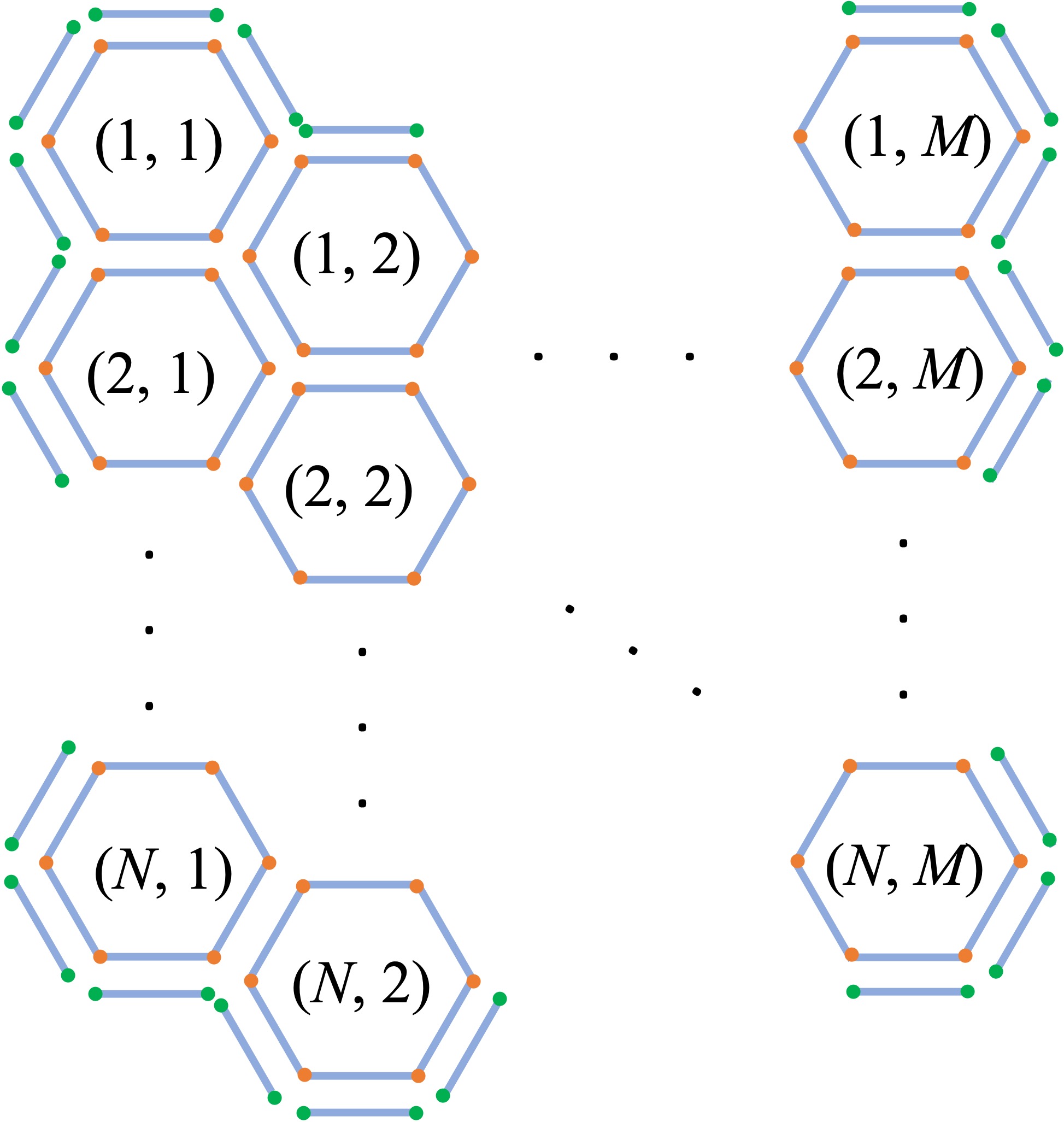}
    \caption{A parallelogram hexagonal mesh of size $N\times M$.}
    \label{fig:hexagonal}
\end{figure}

\begin{figure}[!htb]
    \centering
    \includegraphics[width=0.8\linewidth]{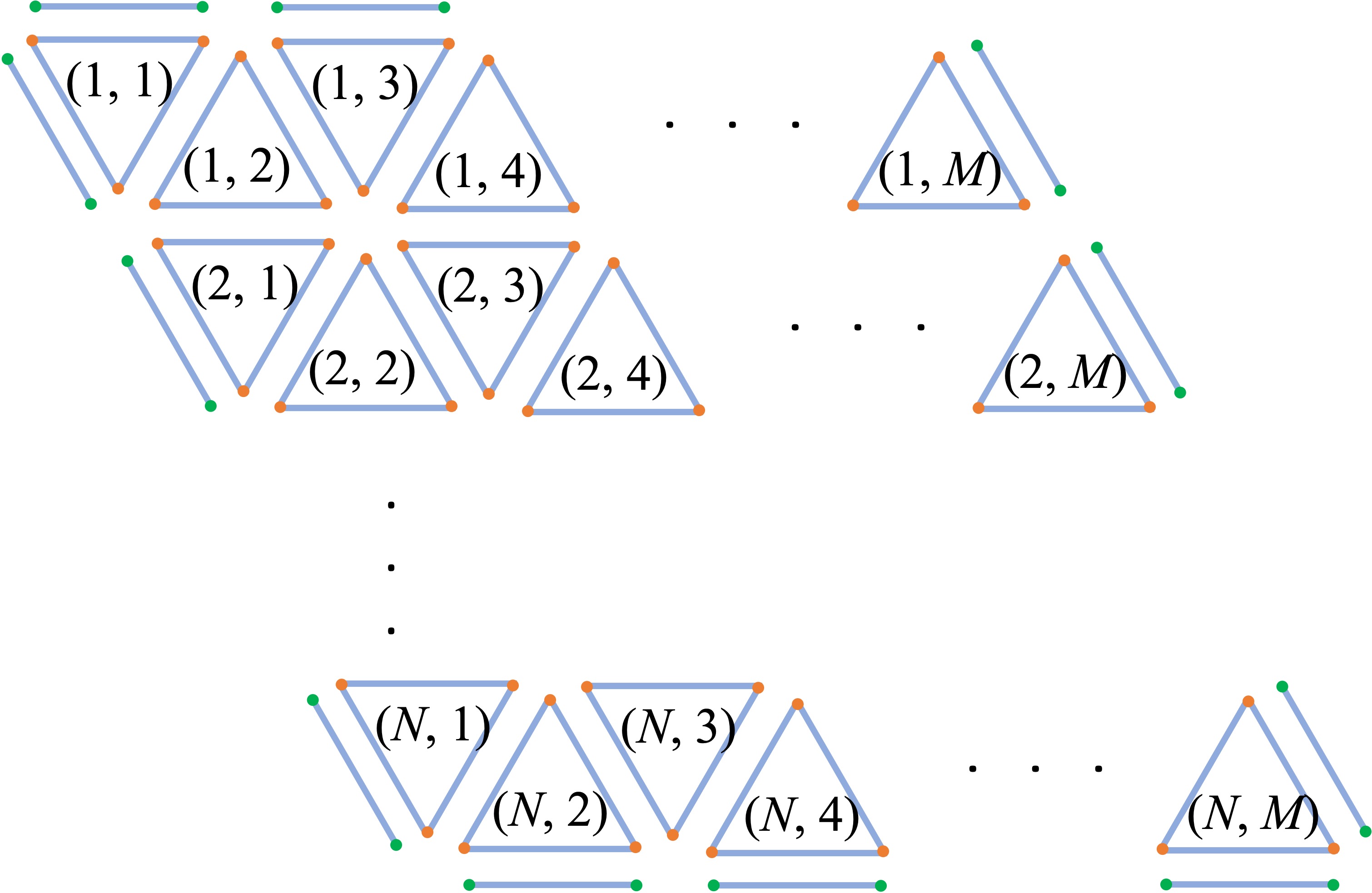}
    \caption{A parallelogram triangular mesh of size $N\times M$. }
    \label{fig:triangular}
\end{figure}

\begin{figure*}[!tb]
    \centering
    \includegraphics[width=0.8\textwidth]{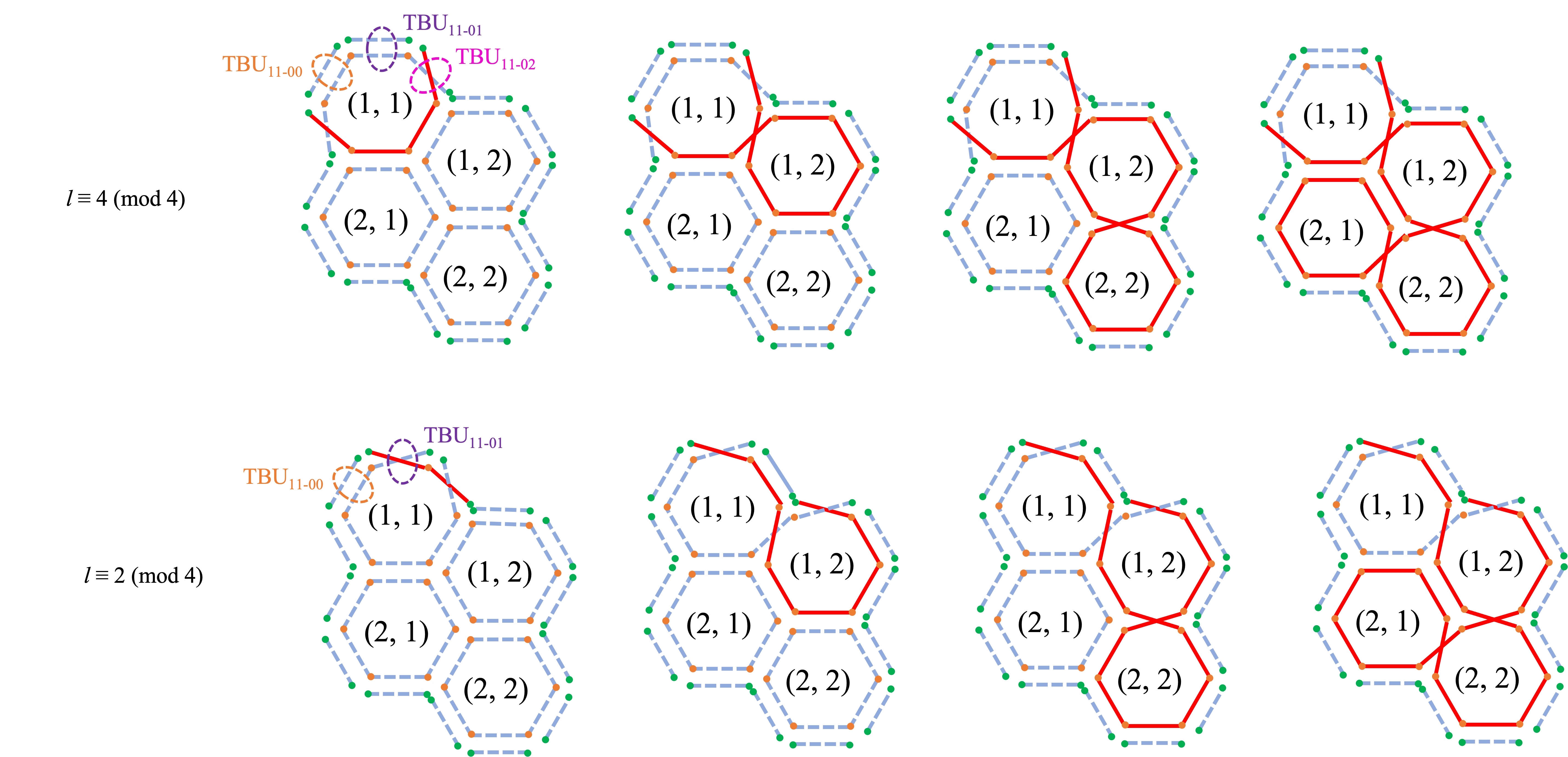}
    \caption{Top row: Construction for finding a circuit configuration to realize $l\equiv 4 \pmod{6}$. Bottom row: Construction for finding a circuit configuration to realize $l\equiv 2 \pmod{6}$. }
    \label{fig:construction_hexagonal}
\end{figure*}

\begin{theorem}\label{thm:hexagonal}
For a $N\times M$ parallelogram hexagonal mesh, we have the following conclusions:
\begin{enumerate}[label = (\arabic*)]
    \item There are $(4N+4M-2)$ peripheral TBUs, and $(3NM-2N-2M+1)$ non-periperhal TBUs; There are $(8N+8M-4)$ floating nodes, and $6NM$ non-floating nodes. There are $(4N+4M-2)$ undirected optical paths in total for one circuit configuration.
    \item Denote an integer set $\Gamma_\star=\{d\in Z^{+}\,|\,1\leq d\leq 6NM+1\}$. Any path length $x\in\Gamma_\star$ is realizable via some circuit configuration, while $x\not\in\Gamma_\star$ is not.
    \item Denote the lengths of all $a$ undirected optical paths by $\Gamma=[l_1,l_2,\cdots,l_{a}]$. Then the sum of all undirected optical paths $\sum_{i=1}^{a} l_i$ can be written in the format $(4N+4M-2+6k)$, for some $k\in\{0,1,\cdots,NM\}$. Moreover, when the path sum equals $(4N+4M-2+6k_0)$, then the path average $\bar{\Gamma}=1+\frac{3k_0}{2N+2M-1}$; the path length $1\leq l_i\leq 6k_0+1$; the path variance $\sigma^2(\Gamma)\leq \frac{18k_0^2}{2N+2M-1}-\frac{9k_0^2}{(2N+2M-1)^2}$.
    \item If $y$ paths each of length $x$ can be realized with some circuit configuration, then we have: 
    $$y \leq \min \{\lfloor\frac{6NM}{x-1}\rfloor, 4N+4M-2\}$$
\end{enumerate}
\end{theorem}

All statements above can be similarly proved following our treatment of the square mesh. In Fig.~\ref{fig:construction_hexagonal}, we demonstrate how to construct a circuit configuration to realize a desired path of length $x\in\Gamma_\star$ in the case of a $2\times 2$ hexagonal mesh. Specifically, the top row and bottom row respectively demonstrate the construction for $l\equiv 4 \pmod{6}$ and $l\equiv 2 \pmod{6}$. To realize $l\equiv 3 \pmod{6}$, we slightly modify the construction method shown in the bottom row by changing the purple $\text{TBU}_{11-01}$ and the orange $\text{TBU}_{11-00}$ to bar and cross state, respectively. Then the resulting red trajectory will have length satisfying $l\equiv 3 \pmod{6}$. The key idea here is the same as shown in Fig.~\ref{fig:construction}. We initially consume cell $(1,1)$ and $(1,2)$ to realize path length $8$ or $9$. Then, when the desired path length increases to $6$, we consume one more cell following a zigzag order: $(1,1),(1,2),\cdots,(1,M),(2,M),\cdots, (2,1),(2,2),\cdots$.  Similarly, the remaining cases $l\equiv 5,0,1 \pmod{6}$ can be handled by modifying the pink, purple, and orange TBUs in the top row. To end this section, Theorem~\ref{thm:triangular} summarizes the conclusions for a parallelogram triangular mesh. The construction method is similar to that for parallelogram hexagonal mesh, and is omitted.

\begin{theorem}\label{thm:triangular}
We assume $M$ is even, so that the $N\times M$ triangular mesh shown in Fig.~\ref{fig:triangular} has a parallelogram shape. We have the following conclusions:
\begin{enumerate}[label = (\arabic*)]
    \item There are $(2N+M)$ peripheral TBUs, and $((3N-1)\frac{M}{2}-N)$ non-periperhal TBUs; There are $(4N+2M)$ floating nodes, and $3NM$ non-floating nodes. There are $(2N+M)$ undirected optical paths in total for one circuit configuration.
    \item Denote an integer set $\Gamma=\{d\in Z^{+}\,|\,1\leq d\leq 3NM+1\}$. Any path length $x\in\Gamma_\star$ is realizable via some circuit configuration, while $x\not\in\Gamma_\star$ is not.
    \item Denote the lengths of all undirected optical paths using a set: $L=[l_1,l_2,\cdots,l_{2N+M}]$. Then the sum of all undirected optical paths $\Gamma=\sum_{i=1}^{2N+M} l_i$ can be written in the format $(2N+M+3k)$, for some $k\in\{0,1,\cdots,NM\}$. Moreover, when the path sum equals $(2N+M+3k_0)$, then the path average $\bar{\Gamma}=1+\frac{3k_0}{2N+M}$; the path length $1\leq l_i\leq 3k_0+1$; the path variance $\sigma^2(\Gamma)\leq \frac{9k_0^2}{2N+M}-\frac{9k_0^2}{(2N+M)^2}$.
    \item If $y$ paths each of length $x$ can be realized with some circuit configuration, then we have: 
    $$y \leq \min \{\lfloor\frac{3NM}{x-1}\rfloor, 2N+M\}$$
\end{enumerate}
\end{theorem}

\end{document}